\documentclass[12pt]{article}

\usepackage{times}
\usepackage{mathrsfs}
\usepackage{amssymb}
\usepackage{amsthm}
\usepackage{amsmath}
\usepackage{subfigure}
\usepackage{graphicx}

\newtheorem{definition}{Definition}[section]
\newtheorem{lemma}{Lemma}[section]
\newtheorem{theorem}{Theorem}[section]

\newtheorem{problem}{Problem}[section]

\newcommand{\RTP}{\textrm{RTP}}
\newcommand{\IRTP}{\textrm{IRTP}}
\newcommand{\SRTP}{\textrm{SRTP}}
\newcommand{\SIRTP}{\textrm{SIRTP}}
\newcommand{\ALGSIRTP}{\textrm{ALGSIRTP}}
\newcommand{\lcm}{\textrm{lcm}}
\newcommand{\poly}{\textrm{poly}}

\begin{document}

\title{Rectangle Transformation Problem
\footnote{This work is supported in part by the 973 Program of China
Grants No. 2014CB340302, and in part by the 973 Program of China
Grants No. 2016YFB1000201.}}

\author{
\normalsize Shaojiang Wang$^{1,3,4}$, Kun He$^{2,3,5}$, Yicheng Pan \footnote{Corresponding author} $\ ^{1,4}$, Mingji Xia$^{1,4}$\\
\small 1. State Key Laboratory of Computer Science,\\ \small Institute of Software, Chinese Academy of Sciences\\ 
\small 2. CAS Key Lab of Network Data Science and Technology,\\ \small Institute of Computing Technology, Chinese Academy of Sciences\\ 
\small 3. University of Chinese Academy of Sciences\\ 
\small 4. Email: \{wangsj, yicheng, mingji\}@ios.ac.cn \\
\small 5. Email: hekun@ict.ac.cn }

\date{}

\maketitle

\begin{abstract}
In this paper, we propose the rectangle transformation problem (RTP)
and its variants. RTP asks for a transformation by a rectangle
partition between two rectangles of the same area. We are interested
in the minimum RTP which requires to minimize the partition size. We
mainly focus on the strict rectangle transformation problem (SRTP)
in which rotation is not allowed in transforming. We show that SRTP
has no finite solution if the ratio of the two parallel side lengths
of input rectangles is irrational. So we turn to its complement
denoted by SIRTP, in which case all side lengths can be assumed
integral. We give a polynomial time algorithm ALGSIRTP which gives a
solution at most $q/p+O(\sqrt{p})$ to $\SIRTP(p,q)$ ($q\geq p$),
where $p$ and $q$ are two integer side lengths of input rectangles
$p\times q$ and $q\times p$, and so ALGSIRTP is a
$O(\sqrt{p})$-approximation algorithm for minimum $\SIRTP(p,q)$. On
the other hand, we show that there is not constant solution to
$\SIRTP(p,q)$ for all integers $p$ and $q$ ($q>p$) even though the
ratio $q/p$ is within any constant range. We also raise a series of
open questions for the research along this line.
\end{abstract}

\section{Introduction}

We consider a practical problem in the Belt and Road initiative
hosted by China. Freight amongst different countries and areas, by
train, ship or plane etc, needs to transfer between transport
facilities. To move cubic boxes between containers, suppose that
each container is a cube, and the source container and target
container usually have different specifications, i.e., different
lengths, widths and heights. A problem is how to design a series of
standard boxes to move between transport facilities of different
specifications easily. For example, a practical challenge of
building the Mongolia, China and Russia economic corridor
infrastructure is their railway gauge differences, in which a key
issue is the container between standard gauge and broad gauge
conversion \cite{L2016}.

Suppose that two containers has the same volume, and we want to
design a series of boxes such that they fully fills in each
container perfectly. \footnote{Usually, two containers are not
exactly of the same volume. However, for some reasons such as
compaction and safety in transportation, the boxes is supposed to be
piled up in a cubic shape. This can be considered equivalent to
transform between two same volumed containers.} When moving from one
container to the other, the number of moving times is hopefully
minimized, which means that the number of boxes is minimized.

A variant of this problem is a simplification by letting each box
have the same height, or say in practice, no cover on the containers
which, however, have the same floor area. So this problem reduces
from $3$-dimensional to $2$-dimensional. Another variant is
prohibiting rotation of each box, which in practice means that no
rotation happens in moving boxes because of machinery constraints.
In this paper, we propose the rectangle transformation problem and
some of its variants to formulate the $2$-dimensional version of
this problem.

\subsection{Definitions and Problems} \label{sec:definition_problem}

We begin with defining rectangle partitions and isomorphic rectangle
partitions as follows.

\begin{definition}
(Rectangle partitions) A \emph{rectangle partition} $\mathcal{P}$ on
a rectangle $M$ is a partition on $M$ such that each module of
$\mathcal{P}$ is a rectangle.
\end{definition}

\begin{definition}
(Isomorphic rectangle partitions) Suppose that $M_1$ and $M_2$ are
two rectangles of the same area. We say that two rectangle
partitions $\mathcal{P}_1$ and $\mathcal{P}_2$ on $M_1$ and $M_2$,
respectively, are \emph{isomorphic} if $\mathcal{P}_1$ and
$\mathcal{P}_2$ (two sets of modules) are exactly the same. That is,
there is a one-one mapping between $\mathcal{P}_1$ and
$\mathcal{P}_2$ such that each pair of modules (smaller rectangles
from $\mathcal{P}_1$ and $\mathcal{P}_2$, respectively) related by
this mapping have the same length and width.
\end{definition}

In the above definition, rotation is allowed in identifying each
pair of modules in the one-one mapping. If rotation is not needed,
then we say that $\mathcal{P}_1$ and $\mathcal{P}_2$ are
\emph{strictly} isomorphic rectangle partitions. The rectangle
transformation problem can be formulated as follows.

\begin{problem}
(Rectangle transformation problem, RTP for short) Let $M_1$ and
$M_2$ be two rectangles $a\times b$ and $c\times d$
($a,b,c,d\in\mathbb{R}^+$), respectively. Suppose that $ab=cd$, that
is, $M_1$ and $M_2$ have the same area. The rectangle transformation
problem requires to find a pair of isomorphic rectangle partitions
$\mathcal{P}_1$ and $\mathcal{P}_2$ for $M_1$ and $M_2$,
respectively.
\end{problem}

The minimum RTP is the optimization problem such that the size (the
number of modules) of $\mathcal{P}_1$ (or of $\mathcal{P}_2$,
equivalently) is minimized. If $\mathcal{P}_1$ and $\mathcal{P}_2$
are required to be strictly isomorphic, then we call the RTP to be
\emph{strict} RTP (SRTP). Given input $a\times b$ and $c\times d$,
the (strict) rectangle transformation problem is essentially
requiring a transformation by a (strict) rectangle partition from
$M_1$ to $M_2$.

Suppose that real numbers $a\geq c\geq d\geq b>0$ and $ab=cd$.
\footnote{We use ``$\geq$" rather than ``$>$" for generality
although $a=c$ means a trivial case.} We formulate the minimum RTP with input
$a\times b$ and $c\times d$ as $\RTP(a,b,c,d)$. Similarly, ignoring
the size relationship between $c$ and $d$, $\SRTP(a,b,c,d)$ can be
defined if we clarify the parallel sides of the two rectangles,
without loss of generality, $a$ and $c$.

An interesting observation for $\SRTP(a,b,c,d)$ is that, when we
shrink a pair of parallel sides, for example, $a$ and $c$, for $d/a$
time, we get two rectangles of size $d\times b$ and $cd/a\times d$.
Since $cd/a=b$, we in fact get two identical rectangles which are
identified by $90^\circ$ rotation. A solution to this new pair of
rectangles implies a solution to the original pair of rectangles by
an $a/d$ time stretch on corresponding sides, and vice versa. Since
SRTP prohibits rotation, the new problem $\SRTP(d,b,b,d)$ is not
easy yet, but it is equivalent to $\SRTP(a,b,c,d)$. Thus, we can omit
two parameters and define $\SRTP(p,q)$ for real numbers $p$ and $q$
to be the SRTP which requires strictly isomorphic rectangle
partitions between rectangles $p\times q$ and $q\times p$.

The above is a general statement of (strict) RTP. In fact, for some
cases, there might be no isomorphic rectangle partitions of finite
size. For example, there are no finite strictly isomorphic rectangle
partitions for $M_1=1\times 2$ and $M_2=\sqrt{2}\times\sqrt{2}$.
Generally, we have the following theorem.

\begin{theorem} \label{thm:irrational_lower_bound}
If $p/q$ is irrational, then there is no solution of finite size to
$\SRTP(p,q)$.
\end{theorem}

We will prove this theorem in Section
\ref{sec:irrational_lower_bound}. In the rest part of this paper, we
focus on the case that $a,b,c,d\in\mathbb{Q}^+$, which turns to be
equivalent to restricting $a,b,c,d\in\mathbb{Z}^+$ for the reason of
stretch technique, where $\mathbb{Z}^+$ denotes the set of positive
integers. In this case, we call RTP to be integral rectangle
transformation problem (IRTP) and SRTP to be strict integral
rectangle transformation problem (SIRTP). Both IRTP and SIRTP have
obviously a trivial solution of size $ab$ (or equivalently $cd$).
But to find an minimum solution is not easy yet. Similar to the
definitions of $\RTP(a,b,c,d)$ and $\SRTP(p,q)$, we also define
$\IRTP(a,b,c,d)$ and $\SIRTP(p,q)$. (Note that after transforming
$\SIRTP(a,b,c,d)$ by the stretch technique, $\SIRTP(d,b,b,d)$ has an
integer input also.) Because of the stretch technique, $\SIRTP(p,q)$
is equivalent to $\SRTP(p,q)$ when $p/q$ is rational, which is the
complemental case of Theorem \ref{thm:irrational_lower_bound}.
Moreover, we can always assume that $a,b,c,d$ are mutually co-prime
for $\IRTP(a,b,c,d)$ and $p,q$ are co-prime for $\SIRTP(p,q)$.

One thing we have to emphasize is that, for $\IRTP(a,b,c,d)$ or
$\SIRTP(p,q)$, the description of a rectangle partition of size at
most $ab$ or $pq$, respectively, might be of super-polynomial length
of the input size. However, the representation of its size which is
an integer at most $ab$ or $pq$ is of polynomial length. So from now
on, we always assume the output of RTP (or IRTP, SRTP, SIRTP) to be
the size of one of the isomorphic rectangle partitions.

\subsection{Relations to Other Partitioning Problems}

It has been known as Wallace-Bolyai-Gerwien theorem for centuries that a polygon can be dissected into any other polygon of the same area.
Precisely, it states that two polygons are equidecomposable in terms of finitely many triangles if and only if they have the same area.
However, for the dissections of a certain shape other than triangle, e.g. for our study, dissecting a rectangle into rectangles, the problem becomes quite different.
Sometimes the equidecomposability is easy for rectangle partitions, but how to find an optimal partition is completely not known, and also rarely studied.

There are several optimization problems about geometrical dissections having been considered.
The most famous one is triangulation, which requires a maximal
partition of the convex hull of a set of points in a plane into
triangles by using these points as triangle vertices. Many
optimization criteria have been studied, for example, optimizing the
minimum or maximum angle \cite{MS1988,ETW1990} and the minimum
weighted triangulation problem \cite{Llo1977,MR2006}.
The minimum-weight triangulation problem asks for a triangulation of a
given point set that minimizes the sum of the edge lengths, and it has been proved to be NP-hard \cite{MR2006}.
Generally, the deterministic version for the minimum number of pieces for polygon transformation,
which is known as $k$-piece dissection problem, has also been proved to be NP-hard \cite{BDDLMRY2015}.

For rectangles, the problem of minimizing the largest perimeter of
modules in rectangle partition of a certain size has been analyzed
\cite{KK1983,AK1992}. Another interesting result states that if a rectangle is partitioned by rectangles
each of which has at least an integer side, then the partitioned rectangle has at least an integer side \cite{W1988}.

However, as far as we know, no optimization
problem on transformations of a certain geometrical figure by basic modules other than triangles has been studied. There are also few algorithmic studies for dissection problems.
Our study opens this research for rectangles.

\subsection{Main Results and Techniques} \label{sec:result_technique}

A straightforward method for a non-trivial solution to
$\IRTP(a,b,c,d)$ ($a\geq c\geq d\geq b$) is using the Euclidean
algorithm (that is, the successive division method) which proceeds
by a greedy heuristic. At the beginning, align two adjacent sides of
both rectangles at a corner arbitrarily, for example, align $a$ with
$c$ and $b$ with $d$. \footnote{The other way to align adjacent
sides is to align $a$ with $d$ and $b$ with $c$ (as shown in Figure
\ref{fig:Euclidean}). Any rule can be raised here to determined
which way should be chosen. But the greedy strategy based on the
fact that $\lfloor a/c\rfloor \leq \lfloor a/d\rfloor$ may not be
always the best rule. For example, consider $\IRTP(15,2,6,5)$.} Then
cut $a$ into $\lfloor a/c\rfloor$ segments and accordingly cut $d$
into $\lfloor d/b\rfloor$ segments (note that $a/c=d/b$). So
$\lfloor a/c\rfloor$ many identical rectangles of size $c\times b$
in both rectangles are identified. The following task is to
transform the unidentified parts $\left(a-\lfloor a/c\rfloor\cdot
c\right)\times b$ and $c\times \left(d-\lfloor d/b\rfloor\cdot
b\right)$ of the same area by isomorphic rectangle partitions, which
is a subproblem of smaller scale. It proceeds recursively until
getting a subproblem of transforming two identical rectangles. The
procedure must end since the algorithm preserves integer side
lengths in each subproblem and the least unit is $1\times 1$ square
in the integer case. The total number of modules is the sum of the
integer parts of the ratios which can be calculated in $O(\log^3 a)$
time since the area of each subproblem halves in each round. So the
Euclidean algorithm halts in polynomial time. The Euclidean
algorithm works for $\SIRTP(p,q)$ also. The only difference is that,
since rotation is not allowed, when align two adjacent sides of both
rectangles, two vertical sides and two horizontal sides are
certainly aligned and successively divided, respectively, and so all
the modules are squares.

\begin{figure}[htbp]
\centering
\includegraphics[scale=.5]{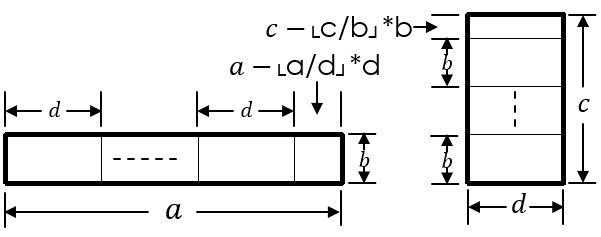}
\caption{Illustration of one step of the Euclidean algorithm.}
\label{fig:Euclidean}
\end{figure}

Note that a trivial lower bound for $\IRTP(a,b,c,d)$ ($a\geq c\geq
d\geq b$) is $\lceil a/c\rceil$ since any rectangle with one side of
length larger than $c$ cannot be contained in the one $c\times d$.
For the same reason, $\lceil q/p\rceil$ is a trivial lower bound for
$\SIRTP(p,q)$ if $q\geq p$. Thus, the Euclidean algorithm with the
greedy strategy in the rule to determine the way of side alignments
gives the optimal solution $a/c$ for $\IRTP(a,b,c,d)$ if $c|a$, and
the optimal solution $q/p$ for $\SIRTP(p,q)$ if $p|q$. For the
general case, we do not yet know whether the Euclidean algorithm
(with a certain rule in determine the alignment of sides in each
round) is a good algorithm for $\IRTP(a,b,c,d)$, but the following
example indicates that it works badly for $\SIRTP(p,q)$. Consider
$\SIRTP(p,p+1)$. In the first step, the Euclidean algorithm divides
the rectangle of size $p\times (p+1)$ (resp. $(p+1)\times p$) into a
square of size $p\times p$ and a slim rectangle of size $p\times 1$
(resp. $1\times p$). But the only way in the subproblem of
transforming $p\times 1$ to $1\times p$ is cutting both of them into
$p$ many $1\times 1$ unit squares since $p$ is a trivial lower bound
for this subproblem. Thus the algorithm gives the solution $p+1$ to
$\SIRTP(p,p+1)$.

In this paper, we focus on SIRTP and establish the algorithm
ALGSIRTP (shown in Section \ref{sec:the_hybrid_algorithm}) which
gives a solution at most $q/p+O(\sqrt{p})$ to $\SIRTP(p,q)$ ($q\geq
p$).
\begin{theorem} \label{thm:algorithm}
Suppose that $p,q\in\mathbb{Z}^+$ and $q\geq p$. ALGSIRTP is a
polynomial time algorithm for $\SIRTP(p,q)$ and gives a solution at
most $q/p+O(\sqrt{p})$.
\end{theorem}
Since $\lceil q/p\rceil$ is a lower bound, if $q=\Omega(p^{3/2})$,
then ALGSIRTP gives a solution to minimum $\SIRTP(p,q)$ with a
constant approximation ratio. In general, it is a
$O(\sqrt{p})$-approximation algorithm.

The main techniques used in ALGSIRTP are a square-transfer heuristic
and its combination with Euclidean. The square-transfer technique
means to build a square to transfer for two parts that are hard to
transform directly. For instance, consider $\SIRTP(p,p+1)$. After
the first step of the Euclidean algorithm, two slim rectangles
$p\times 1$ and $1\times p$ are left after cutting off a $p\times p$
square as a common module. Transforming these two rectangles
directly is costly, but if they can fill in a common part in the
$p\times p$ square simultaneously by partitioning into few modules,
then they can be transferred in this part easily. For the case that
$\sqrt{p}$ is an integer, rectangles $p\times 1$ and $1\times p$ can
be transferred by a $\sqrt{p}\times\sqrt{p}$ square. This heuristic
is illustrated in Section \ref{sec:square_transfer}. The algorithm
in Theorem \ref{thm:algorithm} is a recursive hybrid algorithm of
Euclidean and square-transfer. It is given in Section
\ref{sec:the_hybrid_algorithm}.

Generally, the lower bound $\lceil q/p\rceil$ becomes small if $p$
and $q$ are close. An extreme case is minimum $\SIRTP(p,p+1)$ for
which the Euclidean algorithm gives a $\Omega(p)$-approximation
solution compared with the trivial lower bound $\lceil
(p+1)/p\rceil=2$. For more general lower bound, we have the
following theorem.
\begin{theorem} \label{thm:lower_bound}
For any constants $\varepsilon>0$ and $m\in\mathbb{Z}^+$, there are
positive integers $p$ and $q$ satisfying $p<q<(1+\varepsilon)p$ such
that the minimum solution to $\SIRTP(p,q)$ is more than $m$.
\end{theorem}

Theorem \ref{thm:lower_bound} means that there is not constant
solution independent of $p$ and $q$ to general $\SIRTP(p,q)$. The
main technique used in its proof is to define the pattern of a
rectangle partition. Two pairs of isomorphic rectangle partitions
are equivalent if their two corresponding pairs of rectangle
partitions have the same pattern under a fixed one-one map,
respectively. Then we show a surprising lemma that all the
equivalent pairs can only deal with $\SIRTP(p,q)$ for a fixed value
of $p/q$. Constant size of a partition implies a constant number of
patterns and also a constant number of pattern pairs, which can only
deal with a constant number of values of $p/q$ in the range
$p<q<(1+\varepsilon)p$ for any $\varepsilon>0$. Then Theorem
\ref{thm:lower_bound} follows immediately. We prove it formally in
Section \ref{sec:constant_lower_bound}.

\section{The SIRTP Algorithm}
\label{sec:the_algorithm}

In this section, we give an algorithm which gives a solution at most
$\lfloor q/p\rfloor+8\sqrt{p}+10\log_2 p$ to $\SIRTP(p,q)$ for
$q\geq p$. At the beginning, we illustrate the square-transfer
heuristic by the instance $\SIRTP(p,p+1)$.

\subsection{The Square-Transfer Heuristic}
\label{sec:square_transfer}

Given rectangles $p\times(p+1)$ and $(p+1)\times p$ for which
$\sqrt{p}$ is an integer, we show that there is a pair of isomorphic
partitions of size $2\sqrt{p}+2$. First, we partition these two
rectangles, respectively, into two modules, one of which is a
$p\times p$ square. Then to transform for the left parts $p\times 1$
and $1\times p$, we cut off a square $\sqrt{p}\times\sqrt{p}$ in the
left bottom corner of each $p\times p$ square in both original
rectangles. Then the rectangle $p\times(p+1)$, cut the slim
rectangle $p\times 1$ into $\sqrt{p}$ segments, each of which is
$\sqrt{p}\times 1$, and meanwhile, cut the $\sqrt{p}\times\sqrt{p}$
square into $\sqrt{p}$ many $1\times \sqrt{p}$ rectangles (as
illustrated in Figure \ref{fig:square_transfer_a}). Symmetrically,
for the rectangle $(p+1)\times p$, cut the slim rectangle $1\times
p$ into $\sqrt{p}$ segments, each of which is $1\times\sqrt{p}$, and
meanwhile, cut the $\sqrt{p}\times\sqrt{p}$ square into $\sqrt{p}$
many $\sqrt{p}\times 1$ rectangles (as illustrated in Figure
\ref{fig:square_transfer_b}). Then cutting the rest part of the
module $p\times p$ into two rectangles leads to a pair of isomorphic
rectangle partitions of size $2\sqrt{p}+2$. The key point for which
this can be done is that the two slim rectangles can be cut off into
few parts to fill in a common area (the $\sqrt{p}\times\sqrt{p}$
square) in a large module.

\begin{figure}[htbp]
\centering \subfigure[Rectangle partition on $p\times(p+1)$]{
\label{fig:square_transfer_a}
\includegraphics[width=2.0in]{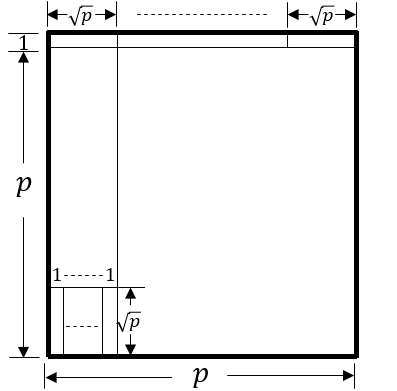}}
\hspace{0.75in} \subfigure[Rectangle partition on $(p+1)\times p$]{
\label{fig:square_transfer_b}
\includegraphics[width=2.0in]{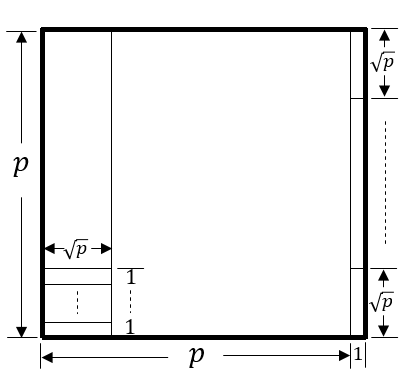}}
\caption{Illustration of the square-transfer heuristic for
$\SIRTP(p,p+1)$ when $\sqrt{p}$ is an integer.}
\label{fig:square_transfer}
\end{figure}

If $\sqrt{p}$ is not an integer, then we first cut each slim
rectangles into $\lfloor\sqrt{p}\rfloor+1$ segments, in which each
of the first $\lfloor\sqrt{p}\rfloor$ segments has length
$\lfloor\sqrt{p}\rfloor$ and width $1$. Then they can be transferred
by a $\lfloor\sqrt{p}\rfloor\times\lfloor\sqrt{p}\rfloor$ square in
the $p\times p$ square. The rest segment of each slim rectangles has
length $p-\lfloor\sqrt{p}\rfloor^2$ and width $1$. Since
$p-\lfloor\sqrt{p}\rfloor^2\leq p-(\sqrt{p}-1)^2=2\sqrt{p}-1$, these
two segments can be transformed directly by at most $2\sqrt{p}-1$
unit squares, which leads to a pair of isomorphic rectangle
partitions of size at most $4\sqrt{p}+1$.

\subsection{The Hybrid Algorithm of Euclidean and Square-Transfer}
\label{sec:the_hybrid_algorithm}

We turn to the general case of $\SIRTP(p,q)$. The idea is to use the
Euclidean step or the square-transfer step recursively for different
cases of subproblems. Note that the idea of square-transfer can be
easily generalized to $\SIRTP(p,p+\Delta)$ if $\Delta$ is far less
than $p$. The method is to choose a proper length of segments for
the less slim rectangles $p\times\Delta$ and $\Delta\times p$ such
that all of them make up of a square. A key observation is that
transforming the left parts of two rectangles is exactly a
subproblem of much smaller area. This is also true for each
Euclidean step in the Euclidean algorithm. So if $\Delta$ is close
to $p$, we use the Euclidean step, and otherwise, we use the
square-transfer step. This recursion process terminates within
$2\lceil\log_2 p\rceil+1$ steps if we guarantee the length of the
shorter side in each subproblem halves in every two rounds.

Next, we state the algorithm firstly, and then show its correctness.

\begin{center}
\begin{minipage}[h]{\textwidth}
\fbox{
\parbox{\textwidth}{
Algorithm $\ALGSIRTP(p,q)$\ \ ($p,q\in\mathbb{Z}^+$)
\begin{enumerate}
\item[] if $p>q$ then $p\leftrightarrow q$
\item[] if $p|q$ then return $\frac{q}{p}$
\item[] else $\Delta\leftarrow(q\mod p)$
\item[] \quad if $\Delta\geq\frac{p}{4}$
\item[] \quad then return
$\left(\left\lfloor\frac{q}{p}\right\rfloor+\left\lfloor\frac{p}{\Delta}\right\rfloor-1
+\ALGSIRTP\left(\Delta,p-\left(\left\lfloor\frac{p}{\Delta}\right\rfloor-1\right)\cdot\Delta\right)\right)$
\item[] \quad else return
\item[] \quad
$\left(\left\lfloor\frac{q}{p}\right\rfloor+2\cdot\left\lfloor\sqrt{\frac{p}{\Delta}-1}\right\rfloor+1
+\ALGSIRTP\left(\Delta,p-\left\lfloor\sqrt{\frac{p}{\Delta}-1}\right\rfloor^2\cdot\Delta\right)\right)$
\item[] \quad endif
\item[] endif
\end{enumerate}
} } \end{minipage}
\end{center}

Note that in each round of recursions, ALGSIRTP calls itself as
subroutine with smaller integer input. So it must terminate after
finite steps. Then we prove Theorem \ref{thm:algorithm}.

\begin{proof}

First, we show that ALGSIRTP is a polynomial time algorithm. Note
that the input length is $\lceil\log_2 p\rceil+\lceil\log_2
q\rceil$. Since division and mod can be calculated in log-square
time, we only have to show that the number of recursions, denoted by
$\ell$, is at most $O(\log p)$, which implies a time complexity
$O(\log q\log p+\log^3 p)$.

We show that the length of the shorter side becomes no more than
half in every successive recursions. In the step stated in
$\ALGSIRTP(p,q)$, note that $\Delta>0$,
$$p-\left(\left\lfloor\frac{p}{\Delta}\right\rfloor-1\right)\cdot\Delta\geq
p-\left(\frac{p}{\Delta}-1\right)\cdot\Delta=\Delta$$ and
$$p-\left\lfloor\sqrt{\frac{p}{\Delta}-1}\right\rfloor^2\cdot\Delta\geq
p-\left(\frac{p}{\Delta}-1\right)\cdot\Delta=\Delta.$$ The length of
the shorter side in each subroutine is $\Delta$. If $\Delta\leq
p/2$, then the length of the shorter side in each subroutine halves.
Otherwise, $\Delta>p/2$ in this round, and thus the first subroutine
$\ALGSIRTP(\Delta,p)$ will be executed in the next round. Now the
value $\Delta':=(p\mod\Delta)=p-\Delta<p/2$. So in the next round,
the length of the shorter side becomes $\Delta'$, which is at most
half of $p$. Therefore, $\ell\leq 2\lceil\log_2 p\rceil+1$.

Next, we show the correctness of the algorithm. Let $p_i$, $q_i$ and
$\Delta_i$ ($0\leq i\leq\ell$) be the parameters in the $i$-th
recursion, where $p_i\leq q_i$, $\Delta_i=(q_i\mod p_i)$, $p_0=p$
and $q_0=q$. We will show that for each $0\leq i<\ell$, the problem
$\SIRTP(p_i,q_i)$ can be reduced to $\SIRTP(p_{i+1},q_{i+1})$ by
adding $\lfloor q_i/p_i\rfloor+\lfloor p_i/\Delta_i\rfloor-1$ many
rectangles if $\Delta_i\geq p_i/4$, and $\lfloor
q_i/p_i\rfloor+2\cdot\lfloor\sqrt{p_i/\Delta_i-1}\rfloor+1$ many
rectangles, otherwise.

For each $0\leq i<\ell$, in the $i$-th round, we cut off $\lfloor
q_i/p_i\rfloor$ many rectangles from both rectangles $p_i\times q_i$
and $q_i\times p_i$ just as what we do in a Euclidean step. When
$\Delta_i\geq p_i/4$, we cut off $\lfloor p_i/\Delta_i\rfloor-1$
more squares $\Delta_i\times\Delta_i$ from both rectangles
$p_i\times \Delta_i$ and $\Delta_i\times p_i$ which are left in the
two rectangles $p_i\times q_i$ and $q_i\times p_i$, respectively.
Then this induces the subproblem $\SIRTP(p_{i+1},q_{i+1})$ where
$p_{i+1}=\Delta_i$ and $q_{i+1}=p_i-(\lfloor
p_i/\Delta_i\rfloor-1)\cdot\Delta_i$, while $\lfloor
q_i/p_i\rfloor+\lfloor p_i/\Delta_i\rfloor-1$ many rectangles are
added. This way of partition is illustrated in Figure
\ref{fig:ALGSIRTP_Delta_large}, where $\lfloor q_i/p_i\rfloor$ is
simplified to be $1$ and the shaded areas shows the subproblem
$\SIRTP(p_{i+1},q_{i+1})$.

\begin{figure}[htbp]
\centering \subfigure[Rectangle partition on $p_i\times q_i$]{
\label{fig:ALGSIRTP_Delta_large_a}
\includegraphics[width=2.0in]{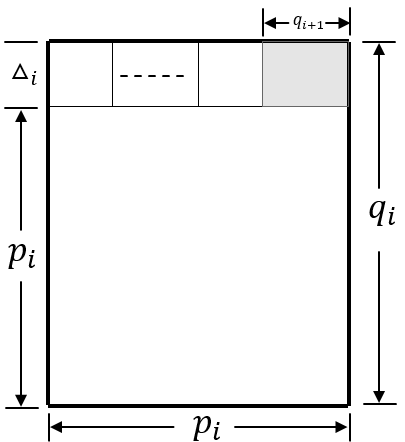}}
\hspace{0.75in} \subfigure[Rectangle partition on $q_i\times p_i$]{
\label{fig:ALGSIRTP_Delta_large_b}
\includegraphics[width=2.15in]{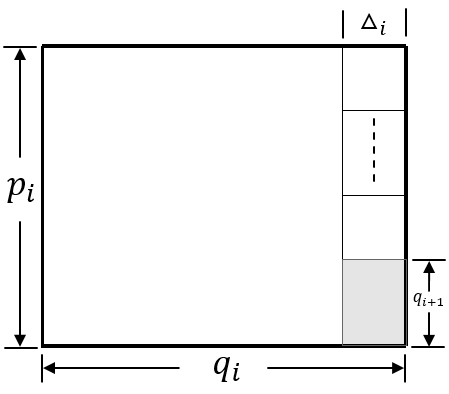}}
\caption{Illustration of one step of recursion for $\SIRTP(p_i,q_i)$
when $\Delta_i\geq p_i/4$.} \label{fig:ALGSIRTP_Delta_large}
\end{figure}

When $\Delta_i<p_i/4$, we use the square-transfer heuristic. We cut
off $\lfloor\sqrt{p_i/\Delta_i-1}\rfloor$ more identical rectangles
$s_i\times\Delta_i$ (resp. $\Delta_i\times s_i$) from $p_i\times
\Delta_i$ (resp. $\Delta_i\times p_i$), where
$s_i=\lfloor\sqrt{p_i/\Delta_i-1}\rfloor\cdot\Delta_i$. These
rectangles make up of the square $s_i\times s_i$, which can be cut
off as a transfer square from one of the larger squares $p_i\times
p_i$ (there must be one). Thus this induces the subproblem
$\SIRTP(p_{i+1},q_{i+1})$ where $p_{i+1}=\Delta_i$ and
$q_{i+1}=p_i-\lfloor\sqrt{p_i/\Delta_i-1}\rfloor^2\cdot\Delta_i$,
while $\lfloor
q_i/p_i\rfloor+2\cdot\lfloor\sqrt{p_i/\Delta_i-1}\rfloor+1$ many
rectangles are added. This is illustrated in Figure
\ref{fig:ALGSIRTP_Delta_small}, where $\lfloor q_i/p_i\rfloor$ is
also simplified to be $1$ and the shaded areas shows the subproblem
$\SIRTP(p_{i+1},q_{i+1})$.

\begin{figure}[htbp]
\centering \subfigure[Rectangle partition on $p_i\times q_i$]{
\label{fig:ALGSIRTP_Delta_small_a}
\includegraphics[width=2.0in]{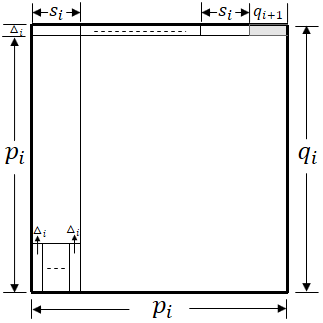}}
\hspace{0.75in} \subfigure[Rectangle partition on $q_i\times p_i$]{
\label{fig:ALGSIRTP_Delta_small_b}
\includegraphics[width=2.0in]{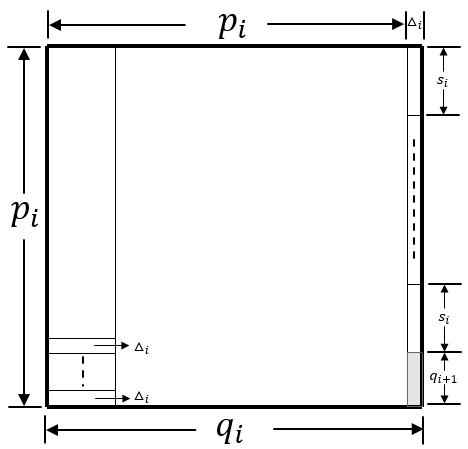}}
\caption{Illustration of one step of recursion for $\SIRTP(p_i,q_i)$
when $\Delta_i<p_i/4$.} \label{fig:ALGSIRTP_Delta_small}
\end{figure}

Then we only have to prove the upper bound in Theorem
\ref{thm:algorithm}. Concretely, we show that the output of
$\ALGSIRTP(p,q)$ is at most $\lfloor q/p\rfloor+8\sqrt{p}+10\log_2
p$ for $q\geq p$. To this end, we prove the following lemma.

\begin{lemma} \label{lem:recursive_upper_bound}
For each $0\leq i\leq\ell$, $\ALGSIRTP(p_i,q_i)\leq \lfloor
q_i/p_i\rfloor+8\sqrt{p_i}+\log_2 p_i+4(\ell-i)$ for $q_i\geq p_i$.
\end{lemma}

\begin{proof}
We prove this lemma by induction on $i$ from $\ell$ down to $0$.	
For the basis case that $p_\ell|q_\ell$,
$\ALGSIRTP(p_\ell,q_\ell)=q_\ell/p_\ell$. The lemma holds for
$i=\ell$. For the inductive step, assume that the lemma holds for
$i=k+1$. Then for $i=k$, we consider the two cases that
$\Delta_k\geq p_k/4$ and $\Delta_k<p_k/4$, respectively.

By the way of partitions, if $\Delta_k\geq p_k/4$, then
$q_{k+1}=p_k-(\lfloor p_k/\Delta_k\rfloor-1)\cdot\Delta_k$ and
$p_{k+1}=\Delta_k<p_k$. We have
\begin{eqnarray*}
&& \ALGSIRTP(p_k,q_k)\\
&\leq&
\left\lfloor\frac{q_k}{p_k}\right\rfloor+\left\lfloor\frac{p_k}{\Delta_k}\right\rfloor-1
+\ALGSIRTP\left(p_{k+1},q_{k+1}\right)\\
&\leq&
\left\lfloor\frac{q_k}{p_k}\right\rfloor+\left\lfloor\frac{p_k}{\Delta_k}\right\rfloor-1
+\left\lfloor\frac{q_{k+1}}{p_{k+1}}\right\rfloor+8\sqrt{p_{k+1}}+\log_2
p_{k+1}+4[\ell-(k+1)]\\
&\leq&
\left\lfloor\frac{q_k}{p_k}\right\rfloor+\left\lfloor\frac{p_k}{\Delta_k}\right\rfloor-1
+\frac{p_k}{\Delta_k}-\left(\left\lfloor\frac{p_k}{\Delta_k}\right\rfloor-1\right)+8\sqrt{p_{k+1}}+\log_2
p_{k+1}+4(\ell-k)-4\\
&<& \left\lfloor\frac{q_k}{p_k}\right\rfloor+8\sqrt{p_k}+\log_2
p_k+4(\ell-k)+\frac{p_k}{\Delta_k}-4\\
&\leq& \left\lfloor\frac{q_k}{p_k}\right\rfloor+8\sqrt{p_k}+\log_2
p_k+4(\ell-k).
\end{eqnarray*}
The lemma holds for $i=k$.

If $\Delta_k<p_k/4$, then
$q_{k+1}=p_k-\lfloor\sqrt{p_k/\Delta_k-1}\rfloor^2\cdot\Delta_k$ and
$p_{k+1}=\Delta_k$. We have
\begin{eqnarray*}
&& \ALGSIRTP(p_k,q_k)\\
&\leq& \left\lfloor
\frac{q_k}{p_k}\right\rfloor+2\cdot\left\lfloor\sqrt{\frac{p_k}{\Delta_k}-1}\right\rfloor+1
+\ALGSIRTP\left(p_{k+1},q_{k+1}\right)\\
&\leq& \left\lfloor
\frac{q_k}{p_k}\right\rfloor+2\cdot\left\lfloor\sqrt{\frac{p_k}{\Delta_k}-1}\right\rfloor+1
+\left\lfloor\frac{q_{k+1}}{p_{k+1}}\right\rfloor+8\sqrt{p_{k+1}}+\log_2
p_{k+1}+4[\ell-(k+1)]\\
&\leq& \left\lfloor
\frac{q_k}{p_k}\right\rfloor+2\cdot\left\lfloor\sqrt{\frac{p_k}{\Delta_k}-1}\right\rfloor+1
+\frac{p_k}{\Delta_k}-\left\lfloor\sqrt{\frac{p_k}{\Delta_k}-1}\right\rfloor^2+8\sqrt{p_{k+1}}+\log_2
p_{k+1}\\
&&+4[\ell-(k+1)]\\
&\leq& \left\lfloor
\frac{q_k}{p_k}\right\rfloor+2\cdot\sqrt{\frac{p_k}{\Delta_k}-1}+1
+\frac{p_k}{\Delta_k}-\left(\sqrt{\frac{p_k}{\Delta_k}-1}-1\right)^2+8\sqrt{\Delta_k}+\log_2
\Delta_k\\
&&+4[\ell-(k+1)]\\
&<& \left\lfloor
\frac{q_k}{p_k}\right\rfloor+4\cdot\sqrt{\frac{p_k}{\Delta_k}-1}+1+8\sqrt{\frac{p_k}{4}}+\log_2
\frac{p_k}{4}+4(\ell-k)-4\\
&<& \left\lfloor\frac{q_k}{p_k}\right\rfloor+8\sqrt{p_k}+\log_2
p_k+4(\ell-k).
\end{eqnarray*}
The lemma holds for $i=k$.

Combining these two cases, Lemma \ref{lem:recursive_upper_bound}
follows.
\end{proof}

Since $\ell\leq 2\lceil\log_2 p\rceil+1$, the upper bound is
obtained when $i=0$ in Lemma \ref{lem:recursive_upper_bound}.
Theorem \ref{thm:algorithm} has been proved.
\end{proof}

By the proof of Theorem \ref{thm:algorithm}, ALGSIRTP gives not only
a solution to $\SIRTP(p,q)$, but also a partition method on the two
rectangles by the choice in each round of recursions, and the
coordinates of partition lines are integers.

Recall that in Section \ref{sec:definition_problem}, we reduce
$\SIRTP(a,b,c,d)$ ($a\geq c\geq d\geq b$) to $\SIRTP(d,b,b,d)$.
Another reduction is to $\SIRTP(a,c,c,a)$ by multiplying $c/b$ to
the side of lengths $b$ and $d$. Since $a/d=c/b$, ALGSIRTP gives the
same solution to $\SIRTP(d,b,b,d)$ and $\SIRTP(a,c,c,a)$ at most
$d/b+O(\sqrt{b})$, and thus the same solution to $\SIRTP(a,b,c,d)$.
The latter one does not necessarily implies partitions with integer
coordinate lines.

For the general case for $\SRTP(p,q)$ that $p/q$ is rational,
because of the stretch technique, we can assume that $p$ and $q$ are
two rational numbers, denoted by $p=p_1/p_2$ and $q=q_1/q_2$ the
irreducible fractions, where $p_i,q_i\in\mathbb{Z}^+$ for
$i\in\{1,2\}$. We can also convert it to the integeral case by
multiplying the least common multiple $\lcm(p_2,q_2)$ to them. Then
assuming $q\geq p$, ALGSIRTP gives a solution at most
$q/p+O(\sqrt{p_1 q_2/g})$, where $g=\gcd(p_2,q_2)$ is the greatest
common divisor of $p_2$ and $q_2$. This is not a
$O(\sqrt{p})$-approximation solution to minimum $\SRTP(p,q)$ and
ALGSIRTP is not be a $\poly(\log p)$ time algorithm any more.
However, in this case, if we are given $p_1,q_1,p_2,q_2$ as input,
then ALGSIRTP halts in $\poly(\log p_2 q_1)$ time.

\section{No Constant Lower Bound}
\label{sec:constant_lower_bound}

In this section, we prove Theorem \ref{thm:lower_bound}. First, we
define the pattern of a rectangle partition and the equivalence of
pairs of isomorphic rectangle partitions. Second, we show that if
two equivalent pairs of isomorphic rectangle partitions solve
$\SIRTP(p,q)$ and $\SIRTP(p',q')$, respectively, then $p/q=p'/q'$.
This means that an equivalent class of isomorphic rectangle
partition pairs can only deal with the $\SIRTP(p,q)$ for a fixed
value of $p/q$. So a constant number of equivalent classes cannot
deal with too many $\SIRTP(p,q)$ instances.

\subsection{The Pattern of A Rectangle Partition}

Suppose that $\mathcal{P}$ is a rectangle partition of size $k$ on
rectangle $M$. An observation is that, except the four corners of
$M$, every the intersection among partition lines and $M$'s sides is
of shape ``$\perp$" or cross. By extending all the partition lines
to be face-to-face, all the ``$\perp$"-shaped intersections become
crosses, while some new cross intersections might emerge. Such
extension makes a grid partition on $M$ and we call it the
\emph{extension of $\mathcal{P}$}. Suppose that in this extension,
there are $r$ rectangles in each row and $c$ rectangles in each
column.

For each integer $1\leq i\leq k$, define matrix $M_i$ to be a
$0$-$1$ matrix of size $r\times c$, in which the $(u\times v)$-th
entry ($1\leq u\leq r$, $1\leq v\leq c$) is $1$ if and only if the
$(u\times v)$-th rectangle in the extension belongs to the $i$-th
module of $\mathcal{P}$. We define the pattern of $\mathcal{P}$ to
be the $k$-tuple $(M_1,M_2,\ldots,M_k)$. We say that two rectangle
partitions of size $k$ have the same pattern if they have exactly
the same $k$-tuple, including the size and the entries of each
$M_i$. Figure \ref{fig:pattern} illustrates the pattern of a
rectangle partition of size $6$, in which each $M_i$ for $1\leq
i\leq 6$ is a $0$-$1$ matrix of size $3\times 4$, and for example,
the second row of $M_3$ is $\begin{pmatrix} 1 & 1 & 1 & 0
\end{pmatrix}$ and other two rows are $\begin{pmatrix} 0 & 0 & 0 & 0
\end{pmatrix}$.

\begin{figure}[htbp]
\centering \subfigure[Original rectangle partition]{
\label{fig:pattern_a}
\includegraphics[width=2.0in]{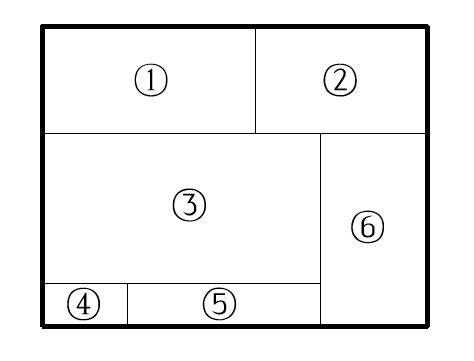}}
\hspace{0.75in} \subfigure[Extension of the original rectangle
partition]{ \label{fig:pattern_b}
\includegraphics[width=2.0in]{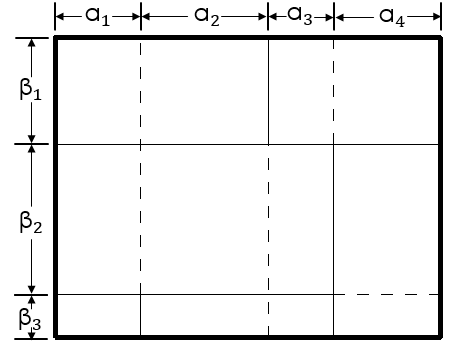}}
\caption{Illustration of the pattern of a rectangle partition.}
\label{fig:pattern}
\end{figure}

Note that for each $i$ the rank of $M_i$ is exactly one since each
module of $\mathcal{P}$ is a rectangle. Denote the non-zero row of
$M_i$ by $\vec{r}_i$ and the non-zero column of $M_i$ by
$\vec{c}_i$. Let $l_i\times h_i$ be the $i$-th module of
$\mathcal{P}$ and $\alpha_u\times\beta_v$ be the $(u\times v)$-th
rectangle in the extension. Let
$\vec{\alpha}_i=(\alpha_1,\alpha_2,\ldots,\alpha_c)$ and
$\vec{\beta}_i=(\beta_1,\beta_2,\ldots,\beta_r)$ be the two length
vectors for the two adjacent sides, respectively. Thus for each
$1\leq i\leq k$,
$$l_i=\langle\vec{\alpha},\vec{r}_i\rangle$$ and $$h_i=\langle\vec{\beta},\vec{c}_i\rangle.$$

\subsection{Proof of Theorem \ref{thm:lower_bound}}

Then we prove Theorem \ref{thm:lower_bound}.

\begin{proof}
For two pairs of isomorphic rectangle partitions
$(\mathcal{P}_1,\mathcal{P}_2)$ and
$(\mathcal{P}'_1,\mathcal{P}'_2)$, suppose that all four partitions
have size $k$. Because of isomorphism, there is a one-one map,
denoted by $\pi$ (resp. $\pi'$)$:[k]\rightarrow[k]$, to identify the
modules in $(\mathcal{P}_1$ and $\mathcal{P}_2)$ (resp. in
$(\mathcal{P}'_1$ and $\mathcal{P}'_2)$). We say that
$(\mathcal{P}_1,\mathcal{P}_2)$ and
$(\mathcal{P}'_1,\mathcal{P}'_2)$ are \emph{equivalent}, if
$\pi=\pi'$, $\mathcal{P}_1$, $\mathcal{P}'_1$ have the same pattern
and $\mathcal{P}_2$, $\mathcal{P}'_2$ have the same pattern also.
For simplification, we can assume that $\pi$ and $\pi'$ are both the
identity map, and each pair of $M_i$'s are equal accordingly.

\begin{lemma} \label{lem:same_ratio}
Suppose that $(\mathcal{P}_1,\mathcal{P}_2)$ and
$(\mathcal{P}'_1,\mathcal{P}'_2)$ are equivalent isomorphic
rectangle partition pairs for $\SIRTP(p,q)$ and $\SIRTP(p',q')$,
respectively ($q\geq p$, $q'\geq p'$). Then $p/q=p'/q'$.
\end{lemma}

\begin{proof}
For $\mathcal{P}_1$, let $\l_i\times h_i$ be the $i$-th module. In
its extension, let $\vec{\alpha}^{(p)}$ be the length vector and
$r^{(1)}_i$ be the non-zero row of $M_i$ parallel to the horizontal
side $p$. Similarly, for $\mathcal{P}'_1$, we define $\l'_i\times
h'_i$, $\vec{\alpha}^{(p')}$, $r^{(1)}_i$ (note that $\mathcal{P}_1$
and $\mathcal{P}'_1$ have the same pattern, and so they have the
same $r^{(1)}_i$); for $\mathcal{P}_2$, we define $\l_i\times h_i$
(note that $\mathcal{P}_1$ and $\mathcal{P}_2$ are isomorphic, and
so the modules are the same), $\vec{\alpha}^{(q)}$, $r^{(2)}_i$; for
$\mathcal{P}'_2$, we define $\l'_i\times h'_i$,
$\vec{\alpha}^{(q')}$, $r^{(2)}_i$. Then by the definition of
pattern, for each $i\in[k]$,
$$l'_i=\langle\vec{\alpha}^{(p')},r^{(1)}_i\rangle=\langle\vec{\alpha}^{(q')},r^{(2)}_i\rangle,$$
Since $$p\cdot q=\sum\limits_{i=1}^k l_i\cdot
h_i=\sum\limits_{i=1}^k
\langle\vec{\alpha}^{(p)},r^{(1)}_i\rangle\cdot h_i,$$ replacing
each $\alpha^{(p)}_i$ by $\alpha^{(p')}_i$ while keeping each $h_i$
makes a rectangle partition of the same pattern on the new rectangle
$p'\times q$. That is,
\begin{equation} \label{eqn:p'_q}
p'\cdot q=\sum\limits_{i=1}^k
\langle\vec{\alpha}^{(p')},r^{(1)}_i\rangle\cdot
h_i=\sum\limits_{i=1}^k l'_i\cdot h_i.
\end{equation}

On the other hand, for the same reason, since
$$q\cdot p=\sum\limits_{i=1}^k l_i\cdot
h_i=\sum\limits_{i=1}^k
\langle\vec{\alpha}^{(q)},r^{(2)}_i\rangle\cdot h_i,$$ replacing
each $\alpha^{(q)}_i$ by $\alpha^{(q')}_i$ while keeping each $h_i$
makes a rectangle partition of the same pattern on the new rectangle
$q'\times p$. That is,
\begin{equation} \label{eqn:q'_p}
q'\cdot p=\sum\limits_{i=1}^k
\langle\vec{\alpha}^{(q')},r^{(2)}_i\rangle\cdot
h_i=\sum\limits_{i=1}^k l'_i\cdot h_i.
\end{equation}
Comparing Equations (\ref{eqn:p'_q}) with (\ref{eqn:q'_p}), we have
$$p'\cdot q=\sum\limits_{i=1}^k l'_i\cdot h_i=q'\cdot p.$$
Lemma \ref{lem:same_ratio} follows.
\end{proof}

By Lemma \ref{lem:same_ratio}, we know that equivalent isomorphic
rectangle partition pairs can only deal with $\SIRTP(p,q)$ for a
fixed ratio $p/q$. Note that for a rectangle partition $\mathcal{P}$
of size $k$, there are at most $k-1$ many face-to-face extensions of
partition lines in each direction. So the number of rectangles in
its extension is at most $k^2$. When $k$ is a constant, the number
of pattern of $\mathcal{P}$ is also a constant. Thus, the number,
denoted by $f(k)$, of equivalent pairs of isomorphic rectangle
partitions of size $k$ is also a constant. However, for any
constants $\varepsilon>0$ and $m\in\mathbb{Z}^+$, there are more
than $f(m)$ many values of $p/q$, which means more than $f(m)$ many
instances of $\SIRTP(p,q)$. By Lemma \ref{lem:same_ratio}, so many
instances cannot be all solved by all pairs of isomorphic rectangle
partitions of size $m$. This completes the proof of Theorem
\ref{thm:lower_bound}.
\end{proof}

\section{Proof of Theorem \ref{thm:irrational_lower_bound}}
\label{sec:irrational_lower_bound}

We define a sort of rectangle partition named slat rectangle
partition that will be used in this proof.

\begin{definition}
(Slat rectangle partition) A \emph{slat rectangle partition}
$\mathcal{P}$ is a rectangle partition on a rectangle $M$, such that
if two horizontal sides of two modules overlap, then they are
identical.
\end{definition}

In other words, when going along any vertical partition line, you
will reach the horizontal sides of $M$, without interruption by any
horizontal side of any module except the horizontal sides of $M$.
The slat rectangle partition, looks like wooden floor composed of
many slats.

\begin{lemma} \label{lem:slat partion}
If $\SRTP(p,q)$ has a solution \footnote{In this section, we resume
the original meaning of a solution to SRTP, that is, a pair of
isomorphic rectangle partitions for two rectangles, rather than the
partition size.}, then there is a solution such that one partition
in the solution is a slat rectangle partition.
\end{lemma}

\begin{proof}
Suppose $\SRTP(p,q)$ has a solution $(\mathcal{P}_1,\mathcal{P}_2)$.
For each rectangle in $\mathcal{P}_1$, we extend its vertical sides.
Since some modules of $\mathcal{P}_1$ are cut into rectangles by the
extended lines,  we get a refinement of $\mathcal{P}_1$, denoted by
$\mathcal{P}'_1$. Obviously, $\mathcal{P}'_1$ is a slat rectangle
partition.

We apply the same refinement to $\mathcal{P}_2$ and get
$\mathcal{P}'_2$. Obviously, $(\mathcal{P}'_1, \mathcal{P}'_2)$ is
also a solution.
\end{proof}

We define the width multi-set of rectangle partition, to be the
multi-set of the widths of its rectangles, and define the width set
to be the set of all widths.

\begin{lemma} \label{lem:width set independt partion}
If $\SRTP(p,q)$ has a solution $(\mathcal{P}_1,\mathcal{P}_2)$, then
there is a solution  $(\mathcal{P}'_1, \mathcal{P}'_2)$ such that
the width set is linear independent over rational numbers. Moreover,
for any $i\in\{1,2\}$, if $\mathcal{P}_i$ is a slat rectangle
partition, then $\mathcal{P}'_i$ is a slat rectangle partition.
\end{lemma}

\begin{proof}
The general intuition of the proof is simple. Whenever there is a
width in the width set that is a linear combination other widths, we
cut all the rectangles of this width vertically into smaller
rectangles according to the combination. While in the proof, we have
to handle the negative coefficients in the linear combinations very
carefully, since all widths are positive, and we can only cut it
into positive segments.

Suppose that the width set of $\mathcal{P}_1$ and $\mathcal{P}_2$ is
$\{x_1,x_2, \ldots, x_n\}$. We can cut all rectangles in
$\mathcal{P}_1$ and $\mathcal{P}_2$ of width $x_i$ into $c$
rectangles of width $x_1/c$. Obviously, this operation, named
mincing operation, changes a slat partition into another one, and
keeps the isomorphism between two partitions.

Suppose that the rank of $\{x_1,x_2, \ldots, x_n\}$ is $r$. The
proof is an induction on $n$. Whenever $n>r$, we construct a new
solution, whose width set has size $n-1$ and rank $r$.

Because $n>r$, the width set is linearly dependent. Notice all
widths are positive numbers. There must be a linear equation of the
form $b_1 y_1+ b_2 y_2 + \cdots + b_s y_s= c_1 z_1 + c_2 z_2 +
\cdots + c_t z_t$, where the coefficients are positive integers, and
$y_1, y_2, \ldots, y_s, z_1, z_2, \ldots, z_t$ are distinct elements
from the width set, and $s, t$ are two positive integers satisfying
$s+t \geq 2$.

We keep $y_1$ unchanged and mince each of other widths $y_2, \ldots,
y_s, z_1, z_2, \ldots, z_t$ into $b_1$ pieces. Because $b_1 y_1+ b_2
b_1 (\frac{1}{b_1}y_2) + \cdots + b_s b_1 (\frac{1}{b_1} y_s)= c_1
b_1 (\frac{1}{b_1} z_1) + c_2 b_1 (\frac{1}{b_1} z_2) + \cdots + c_t
b_1 (\frac{1}{b_1} z_t)$, the new width set satisfies $y_1+ b_2 y'_2
+ \cdots + b_s y'_s= c_1 z'_1 + c_2 z'_2 + \cdots + c_t z'_t$. The
new width set has size $s$ and still rank $r$, since it is linearly
equivalent to the original set.

The purpose of the next step is to mince $y'_2, \ldots, y'_s$ tiny
enough such that they can be put into the bins on the right side of
the equation. There are $c_1$ bins of size $z'_1$, $c_2$ bins of
size $z'_2$ and so on. There are $b_2$ commodities of size $y'_2$,
$b_3$ commodities of size $y'_2$ and so on. The total room of these
bins $c_1 z'_1 + c_2 z'_2 + \cdots + c_t z'_t$ is $y_1$ larger than
total size of the commodities $b_2 y'_2 + \cdots + b_s y'_s$. The
requirement is that each bin in the same size category, leaves the
same room unoccupied.

Let $N$ denote $c_1+c_2+\cdots + c_t$. We always mince the
commodities into pieces of sizes no more than $\frac{y_1}{N}$. If we
have fully utilized a bin such that no piece can be put into it,
then the waste room is no more than $\frac{y_1}{N}$. In this way, we
never waste too much, and the total room of left bins and the room
of unsealed bins, is always larger the size of unpacked commodities.

We mince $y'_2$ into $T c_1$ pieces, where $T$ is larger enough such
that each piece has size no more than $\frac{y_1}{N}$. We put these
pieces into the $c_1$ bins of size $z'_1$ fairly. There are two
cases. If all pieces are put into the $c_1$ bins, we go on to mince
$y'_3$ and put its pieces into bins in the same way, starting from
these unsealed $c_1$ bins.

If there are pieces left, then each bin of size $z'_1$ has a room of
size $z''_2$ unoccupied. We seal these $c_1$ bins, the unoccupied
room of each bin is  no more than $\frac{y_1}{N}$, and it is
unoccupied forever. We go on to put these pieces into the $c_2$ bins
of size $z'_2$. Now, there is a problem, if the number of left
pieces is not a multiple of $c_2$, then we are not able to be fair
to all the $c_2$ bins, and not sure to fulfill the requirement. The
trick is simply mincing each pieces into $c_2$ subpieces, no matter
the packed ones or the unpacked ones. The current size of each
subpieces is $\frac{y'_2}{Tc_1c_2}$

Again, there are two cases, either we pack all subpieces of $y'_2$,
or we seal these $c_2$ bins and begin to pack them into the $c_3$
bins of size $z'_3$. In the second case, we will mince each
subpieces into $c_3$ subsubpieces.

We repeat this process, until all commodities are packed up. Suppose
each $y'_i$ is finally minced into pieces of size $y''_i$,
$i=2,3,\ldots, s$. Suppose the left room of each bin of size $z'_i$
is $z''_i$, $i=1,2,\ldots, t$. We show that $\{y''_2, \ldots, y''_s,
z''_1, \ldots, z''_t\}$ can linearly express $\{y_1, y'_2, \ldots,
y'_s, z'_1, \ldots, z'_t\}$. Each $y'_i$ is a multiple of $y''_i$.
Each $z'_i$ is a linear combination of $z''_i$ and the sizes of
pieces in the corresponding bin. And $y_1$ is a linear combination
of $z''_i$, since it is the size of the total unoccupied room.
Together the rest elements in the width set, we get a smaller new
width set.

The new set has the same rank. We repeat this process, until the
size of the width set becomes $r$.

\end{proof}

Then we turn to prove Theorem \ref{thm:irrational_lower_bound}.

\begin{proof}
Assume that there is a finite solution to $\SRTP(p,q)$.  By Lemma
\ref{lem:slat partion} and \ref{lem:width set independt partion},
there is a solution  $(\mathcal{P}_1,\mathcal{P}_2)$,  such that
$\mathcal{P}_1$ is a slat rectangle partition and the width set
$\{x_1,x_2, \ldots, x_n\}$ is linearly independent over rational
numbers.

Consider the partition $\mathcal{P}_1$. The total height of all
width $x_1$ modules is $sp$, where $s$ is an integer. Then consider
the partition $\mathcal{P}_2$. Suppose that a horizontal line goes
face-to-face. It also goes though some modules, and it is cut into
segments by these modules. The length of each segments is equal to
the width of the rectangle being cut. The total length $p$ is an
integer coefficient linear combination $t_1 x_1+\cdots t_n x_n$ of
$\{x_1,x_2, \ldots, x_n\}$. Because the width set is linearly
independent over rational numbers, the combination realizing $p$ is
unique.

When we move this horizontal line, we always get the same
combination. If we scan the whole rectangle, whose height is $q$,
from top to bottom, we find out that the total height of all width
$x_1$ modules rectangles in partition $\mathcal{P}_2$ is $t_1 q$.
Since $\mathcal{P}_1$ is isomorphic to $\mathcal{P}_2$, $s p = t_1
q$, which contradicts to the fact that $p/q$ is irrational. Theorem
\ref{thm:irrational_lower_bound} has been proved.
\end{proof}

\section{Conclusions and Future Discussions}

In this paper, we proposed the rectangle transformation problem
(RTP), and defined its strict version (SRTP), integral version
(IRTP) and their combination (SIRTP). We showed that $\SRTP(p,q)$
has no finite solution if the ratio of two side lengths $p/q$ is
irrational. So we focused on the complemental case for $\SRTP(p,q)$
that $p$ and $q$ are integers. We gave an algorithm for $\SRTP(p,q)$
which gave a solution at most $q/p+O(\sqrt{p})$ (assume $q\geq p$),
and showed that there is not any constant solution $k$ independent
of $p$ and $q$ for $\SRTP(p,q)$ even if the ratio $q/p$ is in any
constant range.

As a new problem, RTP and its variants leave a lot of open
questions. We list some representative ones for the research along
this line and discuss possible approaches to some of them.

\begin{itemize}

\item[(1)] For $\SIRTP(p,q)$, is there a polynomial time algorithm that gives a
better solution than ALGSIRTP, for example, $q/p+p^\varepsilon$ for
some $\varepsilon<1/2$ or even $q/p+O(\log p)$? We conjecture that
it is possible since the recursive step in ALGSIRTP seems to have a
large development space. By the proof of Theorem
\ref{thm:algorithm}, keeping the total number of recursions within
$O(\log p)$, in each recursive step, additional $p^\varepsilon$ many
rectangles will lead to an improvement to $q/p+p^\varepsilon$ and
additional constant many rectangles will lead to an improvement to
$q/p+O(\log p)$.

\item[(2)] The lower bound for $\SIRTP(p,q)$
seems to be able to be dependent on $p$ and $q$. A straight method
is to improve the proof of Theorem \ref{thm:lower_bound} by finding
an explicit relationship between the rectangle partition size and
the number of patterns. Another question is whether this idea can be
used to IRTP. It is clear that the pattern for IRTP can be defined
similarly, but adjustments to side lengths on a pattern become more
complicated.

\item[(3)] Whether there is a good algorithm for minimum $\IRTP(p,q)$? Here,
``good" means a polynomial time algorithm that have a good
approximation, or even gives the optimal solution. It is a possible
approach to give a rule for the choice of the Euclidean algorithm in
each recursive step as we discussed in Section
\ref{sec:result_technique}.

\item[(4)] For the negative direction of Question (3), what is
the hardness of the deterministic version of minimum $\IRTP(p,q)$
and minimum $\SIRTP(p,q)$, that is, whether there is a solution at
most $k$? If we restrict the partition lines to be of integer
coordinate, a straight method to find the optimal solution is to
enumerate the patterns for all pairs of isomorphic rectangle
partitions along the integer coordinates. Even with this
restriction, the running time of this method is just upper bounded
by $\exp\{O(p,q)\}$. Guessing the patterns and the one-one map for
the isomorphic rectangle partitions implies that the deterministic
versions of minimum IRTP and minimum SIRTP with integral solution
restriction are just both in NEXP. Without this restriction, we do
not even know how to find the optimal solution regardless of running
time. If Question (3) have a positive answer for the optimal
solution, it will be amazing and quite interesting. Note that even
though one of minimum IRTP and minimum SIRTP is easy, it cannot be
implied that the other one is also easy. These two problems are very
different.

\item[(5)] A more practical version of RTP is to relax the target
rectangle to be of area $(1+\delta)S$, where $S$ is the area of the
source rectangle. Then minimum RTP requires to minimize the module
number in the rectangle partition of the source one such that all
modules can be covered by the target one with all sides parallel to
a boundary. This is a mixed scenario about partitioning and
$2$-dimensional bin packing. The latter and its generalization to
$3$-dimensional case have been studied widely for a long time. The
techniques raised there might be helpful. See \cite{LMV2002,P2010}
for a survey.

\item[(6)] RTP and all of its variants can be generalized to
$3$-dimensional version. The two lower bounds in this paper is also
true for $3$-dimensional case, but the algorithm ALGSIRTP cannot be
generalized directly. It is worthwhile to have a systematic study.

\end{itemize}

\end{document}